\documentclass[conference]{IEEEtran}
\usepackage{fancyhdr}
\usepackage{float}
\usepackage{graphicx}
\usepackage{epsfig}
\usepackage{bm}
\usepackage{amsmath}
\usepackage{mathrsfs}
\usepackage{subfigure}
\usepackage{array}
\usepackage{booktabs}
\usepackage{amssymb}
\usepackage{graphicx}
\usepackage{epstopdf}
\usepackage{algorithm}
\usepackage{amsthm}
\usepackage{bbm}
\usepackage{graphicx}
\usepackage{epsfig}
\usepackage{color}
\usepackage{stfloats}
\usepackage{color}
\newtheorem{theorem}{Theorem}

\newtheorem{remark}{Remark}
\makeatother

\author{\IEEEauthorblockN{Chuan Ma$^\star$$^\dag$, Ming Ding$^\dag$, He Chen$^\star$, Zihuai Lin$^\star$, Guoqiang Mao$^\dag$$^\sharp$ and David L\'opez-P\'erez$^\spadesuit$\\
{$^\star$ School of Electrical and Information Engineering, University of Sydney, Australia}\\
{$^\dag$ Data61, CSIRO, Australia}\\
{$^\sharp$ School of Computing and Communications, University of Technology, Sydney, Australia}\\
{$^\spadesuit$ Nokia Bell Labs, Ireland}\\
\small{Email: \{chuan.ma, he.chen, zihuai.lin\}@sydney.edu.au, Ming.Ding@data61.csicro.au, g.mao@ieee.org, david.lopez-perez@nokia.com}}}
\title{On the Performance of Multi-tier Heterogeneous Cellular Networks with Idle Mode Capability}
\begin{document}

\maketitle
\begin{abstract}
This paper studies the impact of the base station (BS) idle mode capability (IMC) on the network performance of multi-tier and dense heterogeneous cellular networks (HCNs). 
Different from most existing works that investigated network scenarios with an infinite number of user equipments (UEs), 
we consider a more practical setup with a finite number of UEs in our analysis. 
More specifically, 
we derive the probability of which BS tier a typical UE should associate to and the expression of the activated BS density in each tier.
Based on such results, 
analytical expressions for the coverage probability and the area spectral efficiency (ASE) in each tier are also obtained. 
The impact of the IMC on the performance of all BS tiers is shown to be significant. 
In particular, 
there will be a surplus of BSs when the BS density in each tier exceeds the UE density, 
and the overall coverage probability as well as the ASE continuously increase when the BS IMC is applied. 
Such finding is distinctively different from that in existing work. 
Thus, 
our result sheds new light on the design and deployment of the future 5G HCNs.
\end{abstract}

\section{Introduction}

Commercial wireless networks are evolving towards higher frequency reuse by deploying smaller cells~\cite{lopez2015towards} to meet the explosively increasing demand for more mobile data traffic~\cite{cisco2015global}.
Heterogeneous cellular networks (HCNs),
which are comprised of a conventional cellular network overlaid with a diverse set of small cell base stations (BSs),
such as metro-, pico- and femto-cells,
can help to realize such view and support much higher data rates per unit area than conventional networks~\cite{7437378}.
It is important to note that each BS tier in a HCN may have different characteristics,
e.g., different spatial density, 
transmit power, 
path loss function, etc.
A comprehensive analysis that takes into account the differences among different BS tiers in a HCN has been carried out in~\cite{jo2012heterogeneous},
where the BS locations are modeled as a homogeneous poison point process (HPPP).
In \cite{jo2012heterogeneous}, 
a flexible user equipment (UE) association strategy was considered,
and each BS tier was assumed to have different spatial densities, 
transmit powers as well as path loss exponents.

The co-channel deployment of macro cell and small cell BSs in HCNs,
i.e., all BS tiers operate on the same frequency spectrum,
have attracted considerable attention recently, 
e.g.,~\cite{6497017},~\cite{7306533}.
Andrews \textit{et al.} in~\cite{andrews2011tractable} first analyzed the coverage probability of a single-tier small cell network by modeling BS locations as a HPPP.
That study concluded that the coverage probability does not depend on the density of BSs in interference-limited scenarios (i.e., when the BSs are dense enough).
Following~\cite{andrews2011tractable},
Jo \textit{et al.} in~\cite{6171996} also reached the same conclusion for each BS tier in a multi-tier HCN.
However, 
it is important to note that the aforementioned work assumed an unlimited number of UEs in the network,
which implies that all BSs would always be active and transmit in all time and frequency resources.
Obviously, 
this may not be the case in practice.

To attain a more practical network performance,
Lee \textit{et al.} in~\cite{lee2012coverage} first analyzed the coverage probability of a single-tier small cell network with a finite  number of UEs,
and derived an optimal BS density accordingly,
by considering the tradeoff between the performance gain and the resultant network cost.
Recently, Ding \textit{et al.} in~\cite{ding2016study} analyzed the coverage probability and area spectral efficiency (ASE) of a single-tier small BS network with probabilistic line-of-sight (LoS) and non-LoS (NLoS) transmissions,
in which the UE number is finite (e.g., 300 UEs/km$^2$) and the small cell BS has an idle mode capability (IMC).
More specifically, 
if there is no active UE within the coverage areas of a certain BS,
that BS will turn off its transmission module using the idle mode.
The IMC switches off unused BSs,
and thus can improve the network energy efficiency and the UEs' coverage probability as the network density increases.
This is because UEs can receive stronger signals from the closer BSs,
while the interference power remains constant thanks to the IMC.
This conclusion in~\cite{ding2016study}
- the coverage probability actually depends on the density of BSs in a interference-limited network -
is fundamentally different from the previous results in~\cite{andrews2011tractable} and~\cite{6171996},
and presents new insights for the design of 5G networks. 
furthermore, 
the IMC even changes the capacity scaling law in ultra-dense networks (UDN)~\cite{Ding2017capScaling}.


However, 
the performance analysis presented in~\cite{lee2012coverage} and~\cite{ding2016study} is only applicable to the single-tier small cell networks.
To our best knowledge, 
the theoretical study of multi-tier and dense HCNs with a finite number of users has not been conducted before,
although some preliminary simulation results can be found in~\cite{lopez2015towards}.

Motivated by the above theoretical gap,
in this paper, 
we for the first time analyze the coverage probability and ASE of a HCN with
\emph{i)} multiple BS tiers,
\emph{ii)} an IMC at small cell BSs, and
\emph{iii)} a limited number of UEs.
To this end, 
we first derive an analytical expression for the density of active BSs in each tier.
Based on this, 
the analytical expressions of the coverage probability and ASE for the HCNs with IMC are obtained.
It is worth pointing out that the extension from a single-tier network scenario to a multi-tier one is \emph{not} trivial,
because BS activation needs to be considered for both the intra-tier and inter-tier BSs.

The rest of this paper is structured as follows.
We describe the system model in~Section~\uppercase\expandafter{\romannumeral2},
and present the main analytical results on the activated BS density, 
the coverage probability and the ASE for each BS tier and for the overall HCN in Section~\uppercase\expandafter{\romannumeral3}.
Numerical results are discussed in Section~\uppercase\expandafter{\romannumeral4}.
Finally, 
the conclusions are drawn in Section~\uppercase\expandafter{\romannumeral5}.

\section{System Model}

We consider a general HCN model consisting of $M$ BS tiers that are characterized by different spatial densities and transmit powers.
The positions of BSs in the $i$-th tier are modeled by a HPPP $\Phi_i$ with a density of $\lambda_i$ BSs/km$^2$.
The positions of UEs are also modeled according to a HPPP $\Phi_{u}$ with a density of $\lambda_{u}$ UEs/km$^2$ that is independent of $\{\Phi_i\}_{i=1,\cdot\cdot\cdot, M}$.
In the majority of the existing works,
$\lambda_{u}$ was assumed to be sufficiently large,
so that each BS in each tier always has at least one associated UE.
However, in our model with finite BS and UE densities,
there may be no UE associated to a BS,
and thus such BS will be turned off thanks to the IMC.

We consider a maximum average received power based cell association strategy,
where each UE associates to only one BS
that provides the maximum average received power.
The average received power from the $i$-th tier is given by
\begin{equation}\label{signal}
	S_i=P_ir^{-\alpha},
\end{equation}
where $P_i$ is the BS transmit power in the $i$-th tier,
$r$ is the distance between the BS and a typical UE sitting at the origin,
and $\alpha$ is the path loss exponent.

Since UEs are randomly and uniformly distributed in the network,
we adopt the following assumption:
the activated BSs in each tier follow a HPPP distribution $\widetilde{\Phi}$,
the density of which is denoted by $\widetilde{\lambda_i}$ BSs/km$^2$~\cite{ding2016study},~\cite{yang2016analysis}.

The SINR of the typical UE with a random distance $r$ from its associated BS in the $i$-th tier is
\begin{equation}\label{SINR}
	\textrm{SINR}_i(r)=\frac{{{P_i}{h_{i0}}{r^{ - \alpha }}}}{{\sum\nolimits_{j = 1}^M {\sum\nolimits_{k \in \widetilde \Phi \backslash {b_0}} {{P_j}{h_{jk}}{{\left| {{Y_{jk}}} \right|}^{ - \alpha }} + {\sigma ^2}} } }},
\end{equation}
where $h_{i0}$ and $h_{ji}$ is the exponentially distributed channel power with unit mean from the serving BS and the $k$-th interfering BS in the $j$-th tier (assuming Rayleigh fading), respectively,
$\left| {{Y_{jk}}} \right|$ is the distance from the activated BS in the $j$-th tier to the origin,
and $b_0$ is the serving BS in the $i$-th tier.
Note that only the activated BSs in $\widetilde \Phi \backslash {b_0}$ inject effective interference into the network,
since the other BSs are muted.

\textbf{It is important to note that it has been shown in~\cite{Ding2017varFactors}
that the analysis of the following factors/models is not urgent, as
they do not change the qualitative conclusions of this type of performance
analysis in UDNs: }\textbf{\emph{(i) }}\textbf{a non-Poisson distributed
BS density, }\textbf{\emph{(ii) }}\textbf{a BS density dependent transmission
power, }\textbf{\emph{(iii) }}\textbf{a more accurate multi-path modeling
with Rician fading, and }\textbf{\emph{(iv) }}\textbf{an additional
modeling of correlated shadow fading. Thus, we will concentrate on
presenting our most fundamental discoveries in this paper, and show
the minor impacts of the above factors/models in the journal version
of this work.}

In Fig. \ref{fig:cell}, 
we show an illustration of the proposed network,
which consists of two BS tiers.
In this case, 
UE~1 connects to BS~2 instead of BS~1 in tier~1 under the assumption that
BS~2 provides the strongest average received signal.
The other BSs are in idle mode since there is no UE associated to them.
\begin{figure}
\centering
  \includegraphics[width=0.35\textwidth,height=5.5 cm]{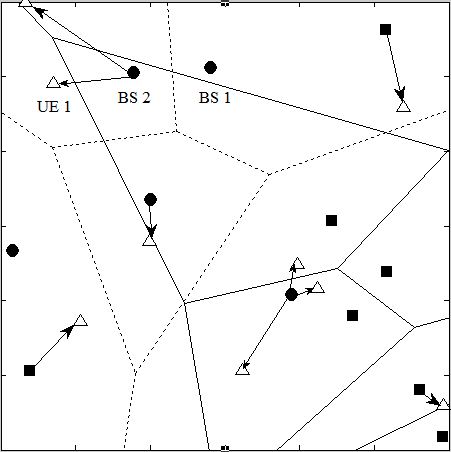}\\
  \caption{Stochastic-geometry model of a two BS tier  HCN.
  The tier 1 BSs, tier 2 BSs and UEs are marked with solid squares, solid dots and empty triangles, respectively.
  Each UE is connected to the BS that provides the strongest average receive signal,
  which is marked by a black arrow.
  BSs with no UE associated are in idle mode. }
  \label{fig:cell}
\end{figure}

\section{Analytical Results}

To evaluate the impact of the IMC on the performance of each BS tier,
we first analyze the probability of a given average number of UEs in each cell.
Then, we derive expressions for the coverage probability and the ASE.

\subsection{Average Number of UEs in each Cell}

The coverage area of each small cell is a random variable $V$,
representing the size of a Poisson Voronoi cell.
Although there is no known closed-form expression for $V$'s probability distribution function (PDF),
some accurate estimates of this distribution have been proposed in the literature, e.g.,~\cite{ferenc2007size} and~\cite{hinde1980monte}.

In~\cite{ferenc2007size},
a simple gamma distribution derived from Monte Carlo simulations was used to approximate the PDF of $V$ for the $i$-th BS tier,
given by
\begin{equation} \label{pdf_of_V_a}
	\begin{split}
		f_{V_i}(x)
= &(b{\lambda _{i}})^q{x^{q - 1}}\frac{{\exp ( - b{\lambda _{i}}x)}}{{\Gamma (q)}},\
	\end{split}
\end{equation}
where $q$ and $b$ are fixed values with $q = b = 3.5$,
$\Gamma(x) = \int^{+\infty}_{0}t^{x-1}e^{-t}dt$ is the standard gamma function
and $\lambda_{i}$ is the BS density of the $i$-th BS tier.

Since the distribution of UEs follows a HPPP with a density of $\lambda_{u}$,
given a Voronoi cell with size $x$,
the number of UEs located in this Voronoi cell is a Poisson random variable with a mean of $\lambda_{u}x$.
Denoting by $N_i$ the number of UEs located in a Voronoi cell in the $i$-th BS tier,
we have
\begin{equation}\label{p_n_uers_in_one_cell}
	\begin{split}
		&\mathbb{P}[N_i = n] \\
		&= \int_0^{ + \infty } {\frac{{{{({\lambda _{u}}x)}^n}}}{{n!}}\exp ( - {\lambda _{u}}x){f_{{V_i}}}(x)dx} \\
		&\overset{(a)}= \frac{{{{(b{\lambda _{i}})}^q}{{({\lambda _{u}})}^n}}}{{n!\Gamma (q){{({\lambda _{u}} + b{\lambda _{i}})}^{q + n}}}}\int_0^{ + \infty } {\exp ( - \xi){\xi^{n + q - 1}}d\xi}\\
		&\overset{(b)}= \frac{{\Gamma (n + q)}}{{\Gamma (n + 1)\Gamma (q)}}{\left(\frac{{{\lambda _{u}}}}{{{\lambda _{u}} + b{\lambda _{i}}}}\right)^n}{\left(\frac{{b{\lambda _{i}}}}{{{\lambda _{u}} + b{\lambda _{i}}}}\right)^q},~~~n\geq0
	\end{split}
\end{equation}
where step (a) is obtained by using $\xi = (\lambda_{u}+b\lambda_{i})x$,
and step (b) is obtained by using the definition of the gamma function.

\subsection{Probability of a UE associated to the $i$-th Tier}
According to (\ref{signal}),
each BS tier's density and transmit power determine the probability that a typical UE is associated with a BS in this tier.
The following theorem provides the per-tier association probability,
which is essential for deriving the main results in the sequel.
\begin{theorem}
The probability of a typical UE associated with a BS in the $i$-th BS tier is given by:
\begin{equation}\label{eq:a1}
	A_i=\frac{{{\lambda _i}}}{{\sum\limits_{j = 1}^M {{\lambda _j}{C_j}^2} }},
\end{equation}
where $i$ denotes the index of the BS tier associating with the typical UE,
and $C_j=({\frac{P_j}{P_i}})^{\frac{1}{\alpha}}$,
where $P_i$ is the BS transmit power of the $i$-th tier,
and $\alpha$ is its path loss exponent.
\end{theorem}
\begin{proof}
See Appendix~A. 
\end{proof}
\vspace{0.2cm}

The intuition of Theorem 1 is that a UE prefers to connect to the BS tier with higher spatial density and transmit power, 
which follows the maximum received power based cell association strategy.

\subsection{Density of activated BSs in the $i$-th tier}
After attaining the probability of one UE associating to a BS in the $i$-th tier,
we are ready to derive the density of activated BS in the $i$-th tier.

Defined by $\mathbb{P}_{i}^{\textrm{off}}(n)$ the probability that the $i$-th tier BS is inactive when there are $n$ UEs in its coverage, then $\mathbb{P}_{i}^{\textrm{off}}(n)$ is given by
\begin{equation} \label{eq:i_n_off}
	\mathbb{P}_{i}^{\textrm{off}}(n)={{{{\mathbb{P} }}[N_i = n] (1 - A_i)} ^n},
\end{equation}
where ${\mathbb{P}}[N_i=n]$ is the probability that $n$ UEs in a cell of $i$-th BS tier and has been obtained from (\ref{p_n_uers_in_one_cell}),
and $A_i$ is the tier association probability obtained by Theorem 1.
\begin{remark}
In (\ref{eq:i_n_off}),
we assume there is no spatial correlation in the UE association process.
Thus, selecting which BS to connect is assumed to be independent for different UEs.
Hence, we have treated the association of the $n$ UEs one by one.
Note that this assumption may cause error as we underestimate the activated BS density of the considered BS tier.
In the simulation section to be presented later, 
we will show that this error is negligible, 
especially when the density of BSs is large.
\end{remark}

The density of activated BSs in the $i$-th tier $\widetilde \lambda _i$ can now be derived as
\begin{equation}
	{\widetilde \lambda _i}=\lambda_i\left(1-\sum\limits^{\infty}_{n=0}\mathbb{P}_{i}^{\textrm{off}}(n)\right),
\end{equation}
where $\mathbb{P}_{i}^{\textrm{off}}(n)$ is the probability that the $i$-th tier is inactivated when there are $n$ UEs in its coverage.

\subsection{The Coverage Probability}
We now investigate the coverage probability that the typical UE's SINR is above a predefined threshold $\tau$.
Since the typical UE is associated with at most one BS,
the coverage probability is given by
\begin{equation}\label{sumcov}
	\mathbb{P}^{\textrm{cov}}=\sum\limits_{i = 1}^M {A_i}\mathbb{E}_r\left\{\mathbb{P}[{\textrm{SINR}}_i(r) > \tau ]\right\},
\end{equation}
where $A_i$ is the probability that the typical UE is associated with the $i$-th BS tier,
which is given by (\ref{eq:a1}) and $\mathbb{P}[\textrm{SINR}_i(r) > \tau ]$ is the coverage probability of the typical UE associated with the $i$-th BS tier.
Our main results on the coverage probability is presented in Theorem~2.

\begin{theorem}
The coverage probability of a typical UE associated with the $i$-th tier is
\begin{equation}\label{cov}
\begin{split}
&\mathbb{E}_r\left\{\mathbb{P}[{\emph{SINR}}_i(r) > \tau ]\right\}\\
&=\int_0^\infty  {\exp \left\{\frac{{ - \tau {r^\alpha }{\sigma ^2}}}{{{P_i}}} - \sum\limits_{j = 1}^M \left( \pi {{\widetilde \lambda }_j}{C_j}^2{r^2}Z(\tau ,\alpha )\right)\right\} f_{r_i}(r)dr},
\end{split}
\end{equation}
with $Z(\tau,\alpha)=\frac{{2\tau }}{{\alpha  - 2}}{}_2{F_1}[1,1 - \frac{2}{\alpha };2 - \frac{2}{\alpha }; - \tau ]$,
and $\alpha>2$ and  $_2F_1[\cdot]$ is the Gauss hypergeometric function.

Moreover,
$f_{r_i}(r)$ is given by
\begin{equation}\label{pdf11}
f_{r_i}(r)=\frac{{2\pi {\lambda _i}r}}{{{A_i}}}\exp \left( - \pi \sum\limits_{j = 1}^M {{\lambda _j}{C_j}^2{r^2}} \right).
\end{equation}
\end{theorem}
\begin{proof}
See Appendix~B.
\end{proof}
\vspace{0.2cm}

By substituting (\ref{eq:a1}) and (\ref{cov}) into (\ref{sumcov}),
we can obtain an analytical expression for the coverage probability.
It is important to note that:
\emph{1)} The impact of the BS tier association and the BS selection on the coverage probability is measured in (\ref{eq:a1}) and (\ref{pdf11}),
the expressions of which are based on $\lambda_i$ and $\lambda_j$.
This is because all the BSs can be chosen by the UEs.
\emph{2)} The impact of the interference on the coverage probability is measured in (\ref{cov}).
Note that instead of $\lambda_j$,
we plug ${\widetilde \lambda }_j$ into (\ref{cov}),
because only the activated BSs emit effective interference into the considered network.

\subsection{Area Spectral Efficiency}

In this subsection, we use the average ergodic rate of a typical UE randomly located in the considered multi-tier network to define the ASE.
Using the same approach as in (\ref{sumcov}),
the average ergodic rate can be expressed as
\begin{equation} \label{sumrate}
\mathbb{R}=\sum\limits_{i=1}^{M}A_i\mathbb{R}_i,
\end{equation}
where $\mathbb{R}_i$ is the average ergodic rate of a typical UE associated with the $i$-th tier,
and it is defined by
\begin{equation}\label{rate1}
\mathbb{R}_i\triangleq \mathbb{E}_r\left\{\mathbb{E}_{\textrm{SINR}_i}[\log_2(1+\textrm{SINR}_i(r))]\right\}.
\end{equation}
The unit of the average rate is bps/Hz/km$^2$.
It is important to note that the average is taken over both the channel fading distribution and the spatial PPP.
The ergodic rate is first averaged on the condition that the typical UE is at a distance $x$ from its serving BS in the $i$-th tier.
Then, the rate is averaged by calculating the expectation over the distance $r$.

We present our results on $\mathbb{R}_i$ in Theorem 3 shown on the top of next page.
By substituting (\ref{eq:a1}) and (\ref{rate}) into (\ref{sumrate}),
we then can obtain an analytical expression for the ASE.
\begin{figure*}[ht]
\begin{theorem}
The average ergodic rate of a typical UE associated with the $i$-th BS tier is
\begin{equation}\label{rate}
\normalsize
\begin{split}
\mathbb{R}_i=
\frac{{2\pi {\lambda _i}}}{{{A_i}}}\int\limits_0^\infty{\int\limits_0^\infty  {r\exp \left\{  - ({2^t} - 1){r^\alpha }{\sigma ^2}{P_i}^{ - 1} - \pi {r^2}\left\{ \sum\limits_{j = 1}^m {{C_j}^2\left[{{\widetilde \lambda }_j}Z({2^t} - 1,\alpha ) + {\lambda _j}\right]} \right\} \right\} dtdr}},
\end{split}
\end{equation}
where $Z({e^t} - 1,\alpha )=\frac{{2({2^t} - 1)}}{{\alpha  - 2}}{}_2{F_1}[1,1 - \frac{2}{\alpha };2 - \frac{2}{\alpha };1 - {2^t}]$, and $\alpha>2$.
\end{theorem}
\begin{proof}
See Appendix~C.
\end{proof}
\hrulefill
\vspace*{4pt}
\end{figure*}

\section{Simulation and Discussion}

In this section, we evaluate the network performance and provide numerical results to validate the accuracy of our analysis.

\subsection{Validation and Discussion on the BS Inactive Probability}

We consider a 3-tier HCN defined by the 3GPP \cite{3Gpp} to show the accuracy of our modeling.
In particular,
we use the following parameter values:
$P_1=46$ dBm, $P_2=30$ dBm, $P_3=24$ dBm, $\lambda_1=10$ BSs/km$^2$, $\lambda_2=100$ BSs/km$^2$ and $\lambda_3 \in[100,500]$ BSs/km$^2$.
Besides,
we adopt the following parameters for the network:
$\alpha=3.75$, $q=b=3.5$,
and the UE density is set to $\lambda_{u}=300$ UEs/km$^2$.
\begin{figure}
\centering
  \includegraphics[width=0.47\textwidth]{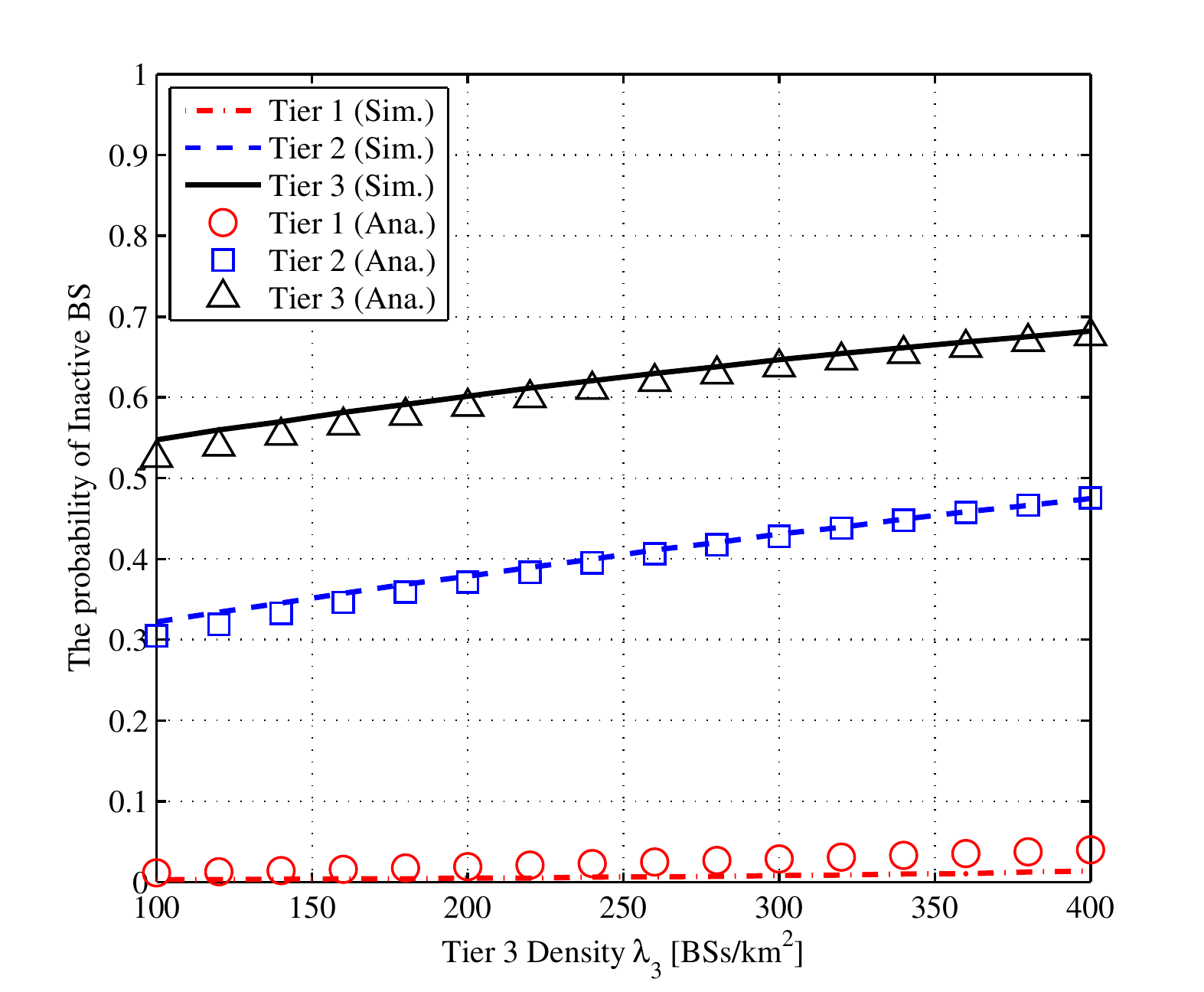}\\
  \caption{The inactive BS probability for each HCN tier}\label{fig:3tieroff}
\end{figure}
In Fig.~\ref{fig:3tieroff},
we draw curves of $\mathbb{P}_i^{\textrm{off}}$ versus $\lambda_3$.
As we can observe from this figure,
our analytical results match well with the simulation results.
Moreover, they also show that
\emph{i)} the probability of a BS being inactive in each tier increases with $\lambda_3$,
when $\lambda_u$ is a finite value,
and that
\emph{ii)} the BSs with a lower transmit power have a less activation probability.
For example, more than 40$\%$ and 60$\%$ of the BSs in tier 2 and tier 3 are idle when $\lambda_3>300~BSs/km^2$.
This means that a large number of UEs are associated with the BSs in tier 1,
as they can provide stronger signals to these UEs.

\subsection{Validation and Discussion on the Coverage Probability}

In this section, we first validate the accuracy of Theorem 2,
where the network consists of 2 BS tiers.
Specially, we use $P_1=30$ dBm, $P_2=24$ dBm, $\lambda_1=100$ BSs/km$^2$ and $\tau=0$ dB.
The rest of the parameters are the same as those in the previous subsection.
\begin{figure}
\centering
  \includegraphics[width=0.47\textwidth]{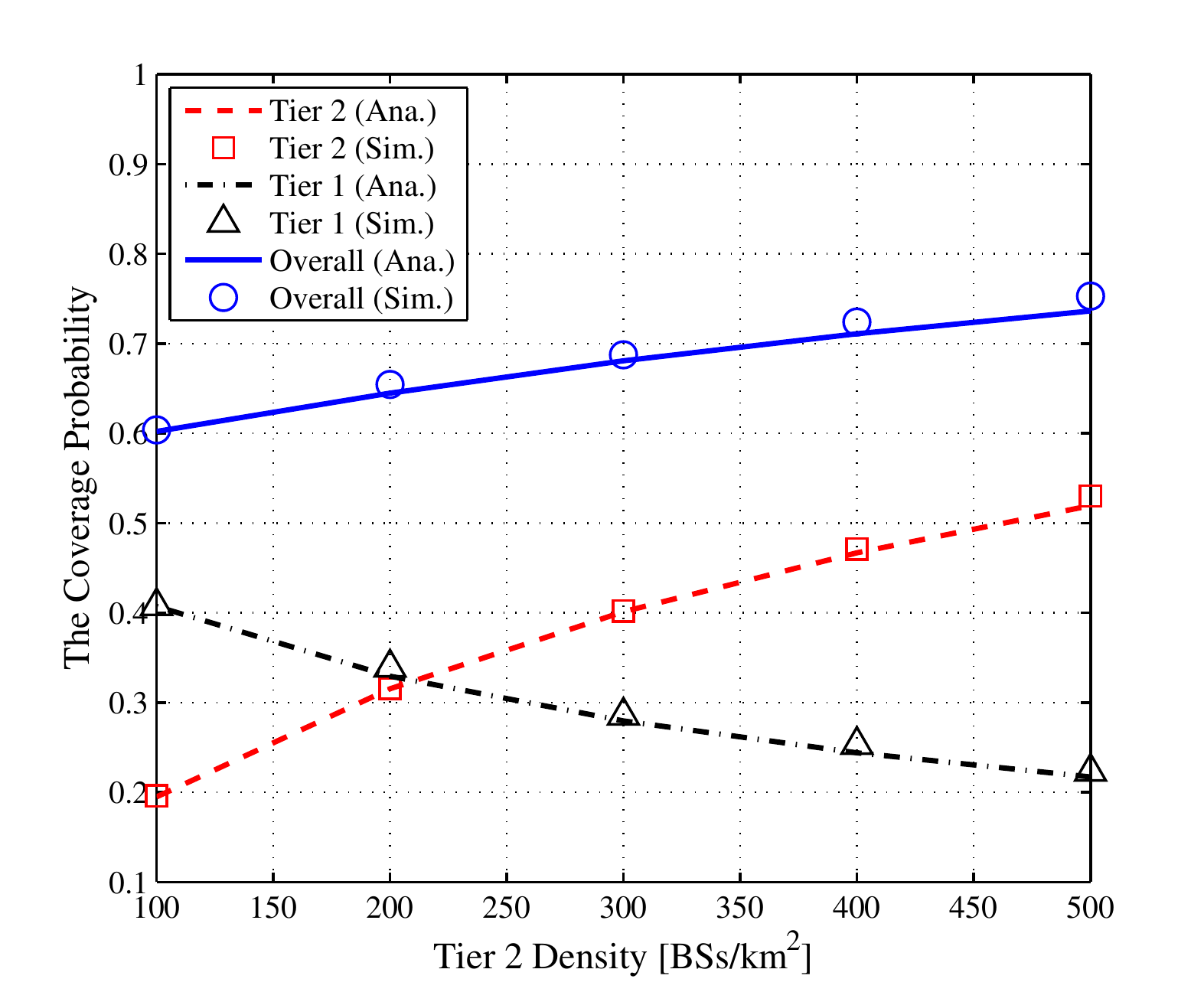}\\
  \caption{The coverage probability with respect to tier 2 density $\lambda_2$}\label{fig:2tiercov}
\end{figure}
In Fig.~\ref{fig:2tiercov},
we show the results of $\mathbb{P}^{\textrm{cov}}$ with respect to $\lambda_1$.
As we can see from the figure,
there are some small errors between  the simulation and  analytical results in each tier.
For example, there is about $1\%$ error when $\lambda_2$ is about 200 BSs/km$^2$.
With the increasing number of BSs,
the error becomes insignificant.
The reason of such error is that the spatial correlation in UE association process is not considered in our analysis.
Specially,
when performing simulations,
nearby UEs have a high probability of being covered and served by the same BS.
However, for tractability, in the analysis,
we consider the BS association of different UEs as independent process,
which underestimates the active BS density.
Due to the good accuracy of $\mathbb{P}^{\textrm{cov}}$,
we will only use analytical results of $\mathbb{P}^{\textrm{cov}}$ for the figures in the sequel.
\begin{figure}
\centering
  \includegraphics[width=0.47\textwidth]{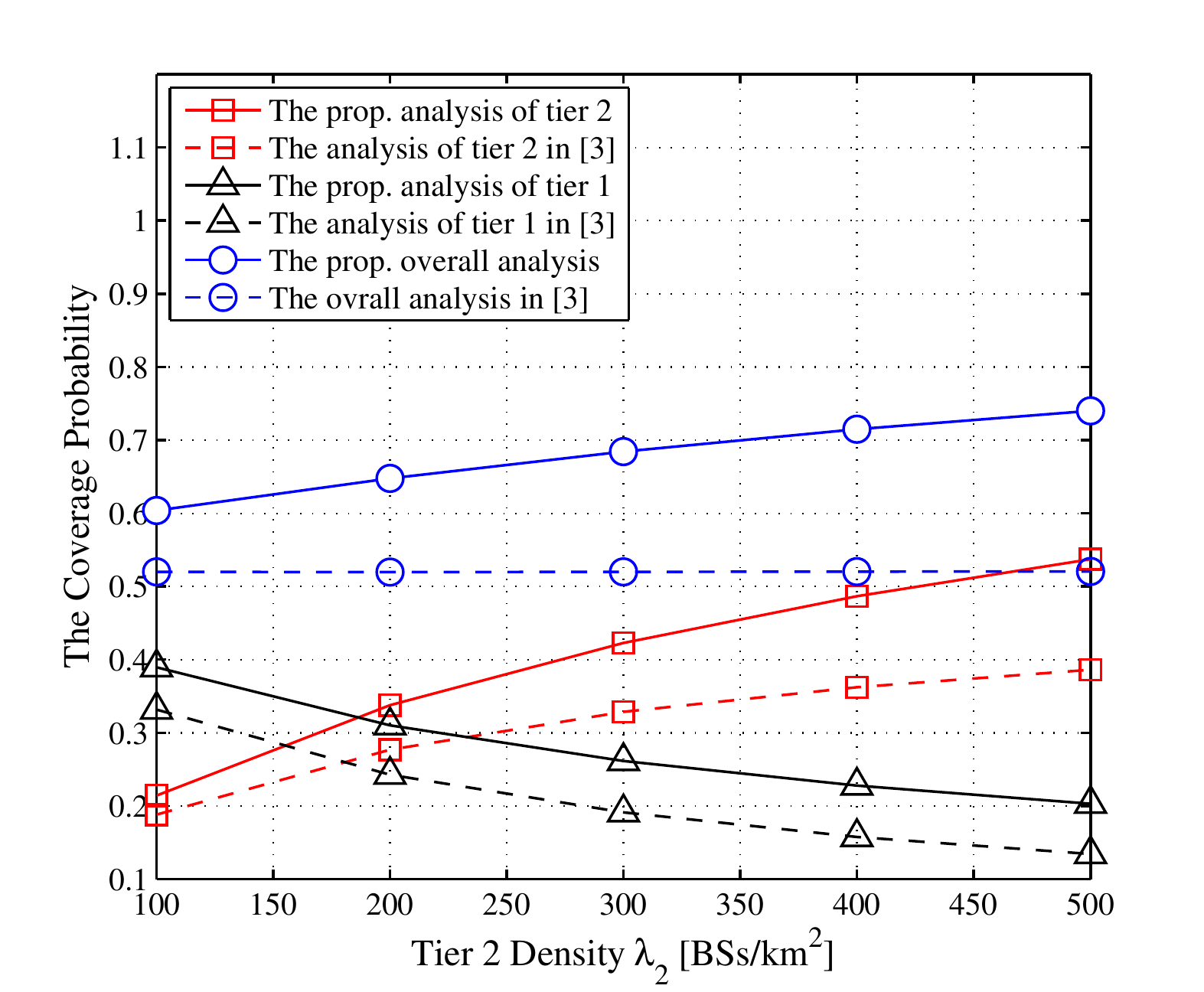}\\
  \caption{The coverage probability compared with the analytical results in \cite{jo2012heterogeneous}}\label{fig:covcompare}
\end{figure}

In Fig.~\ref{fig:covcompare},
we compare our analytical results of coverage probability with those in~\cite{jo2012heterogeneous}.
In~\cite{jo2012heterogeneous},
an infinite number of UEs are considered,
so all BSs are working in the fully-loaded mode.
As can be observed from  Fig.~\ref{fig:covcompare},
based on the results in~\cite{jo2012heterogeneous},
the coverage probabilities of tier~1 and tier~2  decreases and increases as the network densifies, respectively.
As a result, the overall coverage probability approaches a constant.
However, our analytical results show that
although the coverage probabilities of tier~1 and tier~2 show a similar trend as those in~\cite{jo2012heterogeneous}
(also decreases and increases as the network densifies),
the overall coverage probability of the HCN does not follow the same trend.
Due to the IMC considered in our analysis,
the overall coverage probability performance continuously increases as the BS density increases.
The intuition behind this phenomenon is that the interference power will remain constant with the network densification in each tier thanks to the IMC\footnote{The interference power will become constant eventually when there is an increasing number of BSs. Because of the IMC, the number of active BSs is at most equal to the number of UEs, and the distance between one UE and its serving BS keeps decreasing. Thus, from the point of the typical UE, the injecting interference from other active BSs can be regarded as the aggregate interference generated by BSs on a HPPP plane with the same intensity as the UE intensity. Such aggregate interference is bounded and statistically stable [6].},
while the signal power will continuously grow due to the closer proximity of the serving BS as well as the larger BS pool to select from.
This enables stronger serving BS links.

\subsection{Validation and Discussion on the ASE Performance}

\begin{figure}
\centering
  \includegraphics[width=0.47\textwidth]{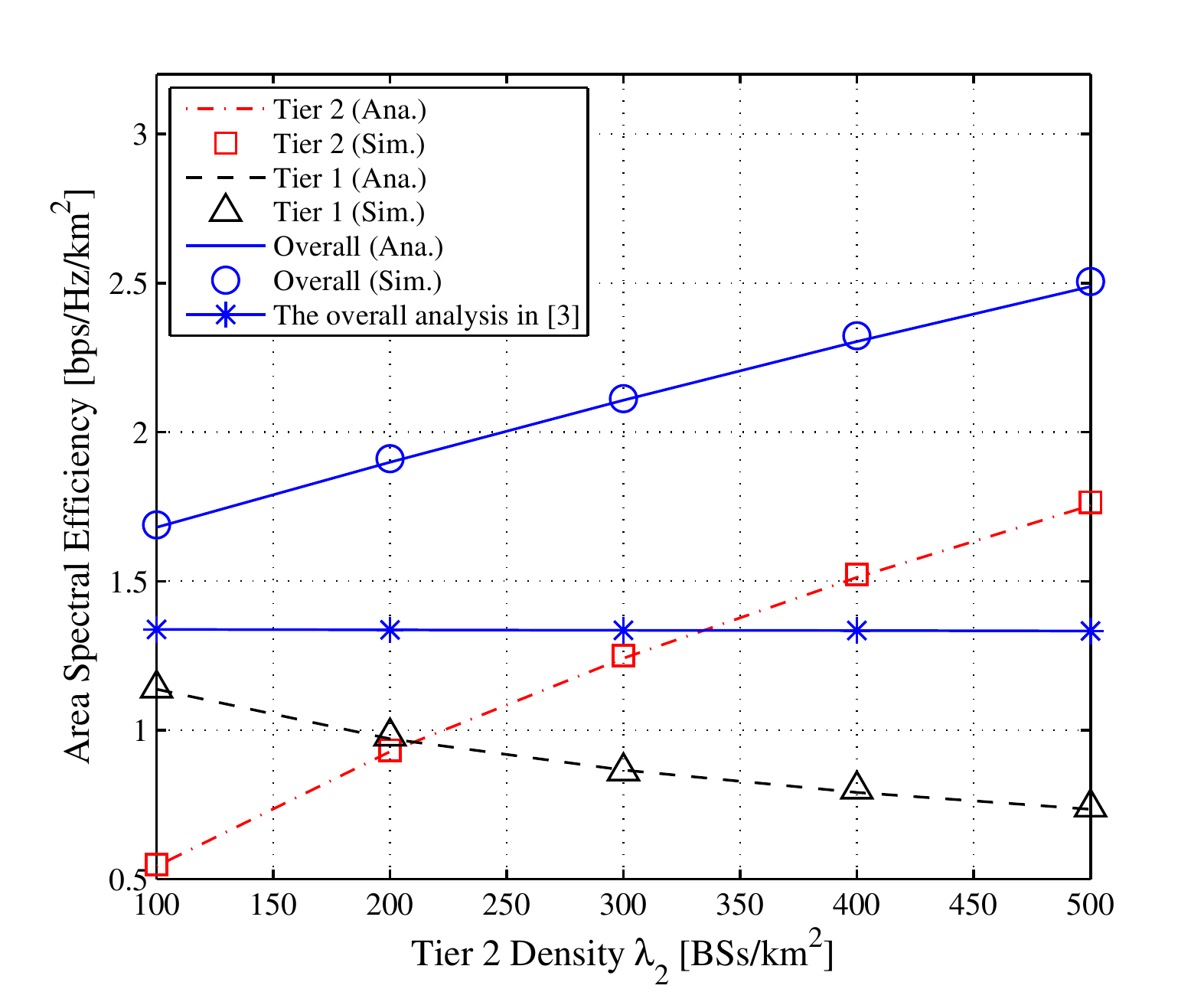}\\
  \caption{The ASE with respect to tier 2 density $\lambda_2$}\label{fig:ase}
\end{figure}

In this subsection, we validate the accuracy of Theorem 3,
As in the previous subsection,
the network consists of 2 BS tiers,
and other parameters are set to be the same with the previous subsections.

In Fig. \ref{fig:ase},
we can observe that the ASE analytical results match well with the simulation results.
Moreover, the results show that with the increasing number of tier 2 BSs,
the ASE of tier 2 increases,
while that of tier 1 reduces.
The performance improvement of tier 2 is because
the interference power from tier 2 BSs remains constant thank to the IMC,
while the UE is served by a stronger link using a BS in the tier 2,
as the network densifies.
However, for the UE associated with tier 1,
although the power of the serving link does not decline,
the cross-tier interference power from tier 2 BSs keeps increasing.
Thus, the performance of tier 1 shows a decreasing trend.
We also compare our proposed overall results with those in~\cite{jo2012heterogeneous}.
Different from~\cite{jo2012heterogeneous},
the impact of IMC is considered in our proposed model,
which results in a better system performance.
The performance of our proposed model grows with the BS density,
while that in \cite{jo2012heterogeneous} is kept constant as the number of BSs increases.
This new observation sheds new light on the design and deployment of HCNs in 5G.

\section{Conclusions}

In this paper, 
we have studied the impact of the idle mode capacity (IMC) on the network performance in multi-tier heterogeneous cellular networks (HCNs) with a limited number of user equipments (UEs). 
It is interesting to observe that the number of tiers and density of base stations (BSs) do affect the cell activation/inactivation probability, 
the coverage probability and the area spectral efficiency (ASE). 
Different from the existing works, 
our results imply that densifying the BSs in each tier will increase the network capacity as well as the quality of service for UEs.

In the considered network model, 
UEs tend to connect to BSs with a larger transmit power, 
e.g., the macrocell BSs, 
and this effect will be aggravated as the power difference among different tiers of BSs becomes larger. 
Thus, 
this leads to a high small cell inactivation probability and an over-utilization of macrocells. 
To avoid this phenomenon, 
the current 4G LTE networks have adopted the technologies of cell range expansion (CRE) and enhanced inter-cell interference coordination (eICIC) combined with the almost blank subframe (ABS) mechanism~\cite{lopez2015towards}. 
As our future work, 
we will consider such mechanisms in our theoretical analyses.

\begin{appendices}

\section{Proof of Theorem 1}

When the typical UE is associated with a BS in the $i$-th tier, the received signal from the serving BS is the largest one, which can be interpreted as:
\begin{equation}\label{eq:event}
\begin{split}
S_{i,0}>\max\limits_{j,j\neq i}S_{j,0}
\Longrightarrow\begin{cases}
S_{i,0}>S_{1,0}\\
S_{i,0}>S_{2,0}\\
...\\
S_{i,0}>S_{j,0}\\
\end{cases}
\Longrightarrow\begin{cases}
r_{1,0}>C_2r_{i,0}\\
r_{2,0}>C_3r_{i,0}\\
...\\
r_{j,0}>C_jr_{i,0},\\
\end{cases}
\end{split}
\end{equation}
where $C_j=\left({\frac{P_j}{P_i}} \right)^\frac{1}{\alpha},~i=2,3,...,m$.

We can calculate the probability of the event in (\ref{eq:event}) as
\begin{equation}\label{eq:multi_condition}
\begin{split}
&\mathbb{P}\left(S_{i,0}>\max\limits_{j,j\neq i}S_{j,0}\right)
=\prod\limits_{j = 1,j\neq i}^M\mathbb{P} \left(r_{j,0} > {C_j}{r_{i,0}}\right)\\
&=\prod\limits_{j = 1,j\neq i}^M \mathbb{P}[\text{No BS closer than $C_j$$r_{i,0}$ in the $j$-th tier}]\\
&=\prod\limits_{j = 1,j\neq i}^M \exp ( - {\pi\lambda_j}{C_j}^2{r_{i,0}^2}).\\
\end{split}
\end{equation}

According to the null probability of a 2-D Poisson process with density $\lambda_i$,
the PDF of $r_{i,0}$ is given by
\begin{equation}\label{ppdf}
	{\mathbb{f}_{{r_{i,0}}}}(r) = \exp ( - \pi {\lambda _i}{r^2})2\pi {\lambda _i}r.
\end{equation}
By combining (\ref{eq:multi_condition}) and (\ref{ppdf}) together, we then have
\begin{equation}
	\begin{split}
		A_i&=\mathbb{E}_{{r_{i,0}}}\left\{\mathbb{P}\left(S_{i,0}>\max\limits_{j,j\neq i}S_{j,0}\right)\right\}\\
		&=\int_0^{\infty}\prod\limits_{j = 1,j\neq i}^M \exp ( - {\pi\lambda_j}{C_j}^2{r^2}){\mathbb{f}_{{r_{i,0}}}}(r)dr\\
		&=\frac{{{\lambda _i}}}{{\sum\limits_{j = 1}^M {{\lambda _j}{C_j}^2} }}.
	\end{split}
\end{equation}

\section{Proof of Theorem 2}
From (\ref{sumcov}), the coverage probability of the $i$-th tier is given by
\begin{equation}
\begin{split}
\mathbb{E}_r\left\{\mathbb{P}[{\textrm{SINR}}_i(r) > \tau ]\right\}=\int_0^\infty\mathbb{P}[{\textrm{SINR}}_i(r) > \tau ]f{_{r_i}}(r)dr,
\end{split}
\end{equation}
where $f{_{r_i}}(x)$ is the PDF of the distance $r_i$ between a typical UE and its serving BS in the $i$-th tier.

Based on the proof in \cite{jo2012heterogeneous}, we can obtain the PDF of $r_i$ as follows,
\begin{equation}\label{pdfff}
f_{r_i}(r)=\frac{{2\pi {\lambda _i}r}}{{{A_i}}}\exp \left( - \pi \sum\limits_{j = 1}^M {{\lambda _j}{C_j}^2{r^2}} \right),
\end{equation}
where $A_i$ is given in Theorem 1.

The SINR of UE in (\ref{SINR}) can be rewritten as $\gamma(r)=\frac{P_ih_{i,0}}{r^\alpha(I_r+\sigma^2)}$, where $I_r=\sum\nolimits_{j = 1}^M\sum\nolimits_{k \in \widetilde \Phi \backslash {b_0}} {{P_j}{h_{jk}}{{\left| {{Y_{jk}}} \right|}^{ - \alpha }}}$. So the CCDF of the typical UE SINR at distance $r$ from its associated BS in the $i$-th tier can be expressed as
\begin{equation}\label{ssss}
\begin{split}
\mathbb{P}[\gamma(r)>\tau]&=\mathbb{P}\left\{h_{i0}>r^\alpha P_i^{-1}(I_r+\sigma^2)\tau\right\}\\
&=\exp\left(\frac{-\sigma^2r^\alpha\tau}{P_i}\right)\prod\limits_{j=1}^M L_{I_r}(r^\alpha P_i^{-1} \tau),
\end{split}
\end{equation}
and the Laplace transform of $I_r$ is
\begin{equation}\label{lap}
\begin{split}
&L_{I_r}(r^\alpha P_i^{-1} \tau)\\
&=\mathbb{E}_{I_r}\left\{\exp\left(-r^\alpha P_i^{-1} \tau I_r\right)\right\}\\
&={E_{{\Phi _j}}}\left\{ \exp \left( - {r^\alpha }{P_i}^{ - 1}\tau {\sum\limits_{k \in {
\widetilde{\Phi} _i}\backslash {B_{i,0}}} {{P_j}{h_{jk}}\left| {{Y_{jk}}} \right|} ^{ - \alpha }}\right)\right\}\\
&\overset{(a)}=\exp \left\{  - 2\pi {\widetilde \lambda _j}\int_{{C_j}r}^\infty  {\left[1 - {L_{{h_{jk}}}}({r^\alpha }{C_j}{y^{ - \alpha }})\right]} ydy\right\}\\
&=\exp \left\{  - 2\pi {\widetilde \lambda _j}\int_{{C_j}r}^\infty  {\frac{y}{{1 + {{({r^\alpha }{C_j}\tau )}^{ - 1}}{y^\alpha }}}dy} \right\}\\
&\overset{(b)}=\exp \left\{  - \pi {\widetilde \lambda _j}{C_j}^2{r^2}{\tau ^{\frac{2}{\alpha }}}\int_{{\tau ^{ - \frac{2}{\alpha }}}}^\infty  {{{(1 + {u^{\frac{\alpha }{2}}})}^{ - 1}}du} \right\} \\
&\overset{(c)}=\exp \left\{  - \pi {\widetilde \lambda _j}C_j^2{r^2}Z(\tau ,\alpha )\right\}
\end{split}
\end{equation}
where step (a) states that the closest interferer in the $j$-th tier is at least at a distance $C_jr$, step (b) is obtained from $u=(x^\alpha C_j \tau)^{-\frac{2}{\alpha}}y^2$, and $Z(\tau,\alpha)=\frac{{2\tau }}{{\alpha  - 2}}{}_2{F_1}[1,1 - \frac{2}{\alpha };2 - \frac{2}{\alpha }; - \tau ]$, and $\alpha>2$ and  $_2F_1[\cdot]$ denotes the Gauss hypergeometric function in step (c).
Combining (\ref{pdfff}), (\ref{ssss}) and (\ref{lap}), we obtain the coverage probability of a typical UE associated with the $i$-th tier in (\ref{cov}).

\section{Proof of Theorem 3}
From (\ref{rate1}), the average ergodic rage of the $i$-th tier is
\begin{equation}
\mathbb{R}_i= \int_0^{\infty}\left\{\mathbb{E}_{\textrm{SINR}_i}[\log_2(1+\textrm{SINR}_i(r))]\right\}f_{r_i}(r)dr,
\end{equation}
where $f_{r_i}(x)$ is given by (\ref{pdfff}). Since $\mathbb{E}[R]=\int_0^{\infty}\mathbb{P}[R>r]dr$ for $R>0$, we can obtain
\begin{equation}
\begin{split}
&\mathbb{E}_{\textrm{SINR}_i}[\log_2(1+\textrm{SINR}_i(r))]\\
&=\int_0^{\infty}\mathbb{P}\left\{\log_2[1+\textrm{SINR}_i(r)]>t\right\}dt\\
&=\int_0^{\infty}\mathbb{P}\left(\textrm{SINR}_i(r)>2^t-1\right)dt
\end{split}
\end{equation}
The rest proof is similar with Appendix A, and the result is obtained by plugging $\tau = 2^t -1$.
\end{appendices}
\bibliographystyle{IEEEtran}
\bibliography{Cellonoff}
\end{document}